\documentclass[11pt]{article}

\usepackage{amsthm,amsmath,amssymb,amsfonts,graphicx,multicol,multirow}
\usepackage[small,it]{caption}
\usepackage{epsfig}
\usepackage{epstopdf}
\usepackage{enumerate}

\usepackage{fullpage}
\usepackage[boxed,linesnumbered]{algorithm2e}
\usepackage[sanserif]{complexity}
\usepackage{graphics}
\usepackage{tikz}

\newclass{\coUL}{coUL}
\newclass{\TISP}{TISP}
\newclass{\ReachUL}{ReachUL}
\newclass{\ReachFewL}{ReachFewL}
\newclass{\ReachLFew}{ReachLFew}
\newclass{\FewL}{FewL}
\newclass{\StrongUL}{StrongUL}
\newclass{\StrongFewL}{StrongFewL}
\newclass{\StrongLFew}{StrongLFew}
\newclass{\PromiseMinUnique}{PromiseMinUnique}
\newclass{\MinUnique}{MinUnique}
\newclass{\PromiseMangroves}{PromiseMangroves}
\newclass{\Mangroves}{Mangroves}
\newclass{\nspaceamb}{NspaceAmbiguity}
\newclass{\ruspace}{RUSPACE}
\newclass{\reachuspace}{ReachUSPACE}
\newclass{\ShP}{ShortestPath}
\newclass{\RB}{RedBluePath}
\newclass{\EB}{EvenPath}
\newclass{\PM}{PerfectMatching}
\newclass{\HO}{HallObs}
\newclass{\ExPM}{ExactPM}
\newclass{\EPM}{EvenPM}

\newcommand{\comment}[1]{}
\newcommand{\ssize}{n^{\frac{1}{2}+\epsilon}}
\newcommand{\sspace}{n^{\frac{1}{2}+\epsilon}}

\newcommand{\tildeO}{\widetilde{O}}

\newcommand{\algoPlanarReach}{{\tt DirectedPlanarReach}}

\newcommand{\PlSep}{\tt PlanarSeparator}
\newcommand{\PlSepFam}{{\tt PlanarSeparatorFamily}}
\newcommand{\PlShortPath}{{\tt PlanarDist}}
\newcommand{\PlShortReport}{{\tt PlanarShortPath}}
\newcommand{\BellmanFord}{{\tt Bellman-Ford}}
\newcommand{\PlOddCycle}{{\tt UPlanarOddCycle}}

\SetKwInOut{Input}{Input}
\SetKwInOut{Output}{Output}
\SetKwInOut{Promise}{Promise}

\usepackage{epsfig}
\usepackage{url}

\newtheorem{theorem}{Theorem}
\newtheorem{lemma}{Lemma}[section]
\newtheorem{claim}{Claim}[section]
\newtheorem{proposition}{Proposition}[section]
\newtheorem{observation}{Observation}
\newtheorem{corollary}{Corollary}[section]

\newtheorem{definition}{Definition}

\newcommand{\ignore}[1]{}

\newcommand{\newfontobj}[2]{
  \newcommand{#1}[1]{
    \expandafter\def\csname##1\endcsname{{#2 ##1}}}}

\comment{

\newcommand{\poly}{\mbox{\tt poly}}

\class{RL} \class{BPL}

\class{PSPACE}          
\class{AM}              
\class{AMTIME}
\class{NC}
\class{AMSUBEXP}
\class{MA}
\class{NP}
\class{NLIN}
\class{L}
\class{BPSPACE}
\class{LINSPACE}
\class{P}
\class{RP}
\class{BPP}
\class{RNC}
\class{PrAM}
\class{NPSV}
\class{advBPP}
\class{DTIME}
\class{E}
\class{ZPP}
\class{EXPSPACE}
\class{SUBEXP}
\class{PH}
\class{IP}
\class{SNP}
\class{SNQP}
\class{AVP}
\class{QP}
\class{SNTIME}
\class{DSPACE}
\class{BPSPACE}
\class{BPTIME}

\class{DTIME}

\class{AvgTIME} \class{DistNP} \class{AZPP} \class{AVE}
\class{DistE} \class{DistEXP} \class{E} \class{NE} \class{UE}
\class{EXP} \class{pcomp} \class{NEXP} \class{UP} \class{DTIME}
\class{NTIME} \class{R} \class{NSUBEXP} \class{SUBEXP} \class{PF}
\class{LINcomp}
\class{BPTIME}
\class{DistNLIN}
\class{AME}
\class{SIZE}
\class{NSIZE}
\class{SAT}
\class{SIZE}

}

\begin{document}
\title{Simultaneous Time-Space Upper Bounds for Certain Problems in Planar Graphs\thanks{Research supported by Research-I Foundation}\footnote{A preliminary version was accepted in WALCOM 2015}}



\author{Diptarka Chakraborty\thanks{Indian Institute of Technology, Kanpur {\tt diptarka@cse.iitk.ac.in}} \and  Raghunath Tewari\thanks{Indian Institute of Technology, Kanpur {\tt rtewari@cse.iitk.ac.in}} }

\maketitle

\begin{abstract}
In this paper, we show that given a weighted, directed planar graph $G$, and any $\epsilon >0$, there exists a polynomial time and $O(\sspace)$ space algorithm that computes the shortest path between two fixed vertices in $G$.

We also consider the {\RB} problem, which states that given a graph $G$ whose edges are colored either red or blue and two fixed vertices $s$ and $t$ in $G$, is there a path from $s$ to $t$ in $G$ that alternates between red and blue edges. The {\RB} problem in planar DAGs is {\NL}-complete. We exhibit a polynomial time and $O(\sspace)$ space algorithm (for any $\epsilon >0$) for the {\RB} problem in planar DAG.

In the last part of this paper, we consider the problem of deciding and constructing the perfect matching present in a planar bipartite graph and also a similar problem which is to find a Hall-obstacle in a planar bipartite graph. We show the time-space bound of these two problems are same as the bound of shortest path problem in a directed planar graph.
\end{abstract}

\section{Introduction}
\label{sec:intro}
Computing shortest path between two vertices in a weighted, directed graph is a fundamental problem in computer science. There are several popular and efficient algorithms that are known for this problem such as Dijkstra's algorithm \cite{Dijkstra59} and Bellman-Ford algorithm \cite{Bellman58} \cite{Ford56}. Both of these algorithms require linear amount of space and run in polynomial time. However Bellman-Ford algorithm is more versatile since it is also able to handle graphs with negative edge weights (but no negative weight cycles). There is also a more recent algorithm by Klein, Mozes and Weimann \cite{KMW09} which runs in polynomial time (with better parameters) but still requires linear space, however this algorithm considers shortest path problem only for directed planar graphs.

Another fundamental problem in space complexity theory that is closely related to the shortest path problem is the problem of deciding reachability between two vertices in a directed graph. This problem characterizes the complexity class non-deterministic logspace or {\NL}. Savitch \cite{Sav70} showed that {\NL} is contained in $\L^2$ ({\L} is {\em deterministic log-space} class), however Savitch's algorithm takes $\Theta(n^{\log n})$ time. Barnes et. al. \cite{BBRS92} gave a $O(n/2^{k\sqrt{\log n}})$ space, polynomial time algorithm for this problem. It is an important open question whether we can exhibit a polynomial time and $O(n^{1-\epsilon})$ space algorithm for the reachability problem in directed graphs, for any $\epsilon >0$. The readers may refer to a survey by Wigderson \cite{widgerson-survey} to know more about the reachability problem. Imai et. al. \cite{INPVW13} answered this question for the class of directed planar graph. They gave a polynomial time and $O(\sspace)$ space algorithm by 
efficiently constructing a {\em planar separator} and applying a divide and conquer strategy. In a recent work, their result has been extended to the class of {\em high-genus} and {\em $H$-minor-free} graphs \cite{CPTVY14}. 

The natural question is whether we can extend these results to the shortest path problem. For a special class of graphs known as {\em grid graphs} (a subclass of planar graphs), Asano and Doerr devised a $O(n^{\frac{1}{2}+\epsilon}) $ space and polynomial time algorithm for the shortest path problem \cite{AD11}. In their paper, Asano and Doerr posed the question whether their result can be extended to planar graphs in general. In this paper, we give a positive answer to their question and exhibit the first sub-linear space, polynomial time algorithm for the shortest path problem in planar graphs. Note that the shortest path problem for both undirected and directed graph is {\NL}-complete \cite{Gold08}. 

Another interesting generalization of the reachability problem, is the {\RB} problem (for the definition see Section \ref{sec:redblue}). The {\RB} problem is {\NL}-complete even when restricted to planar DAGs \cite{Kul09}. A natural relaxation of the above problem is {\EB} problem defined in Section \ref{sec:redblue}. In general, {\EB} problem is {\NP}-complete \cite{LP84}, but for planar graphs, it is known to be in {\P} \cite{Ned99}. In this paper, we also give the first sublinear space and polynomial time algorithm known for the {\RB} and the {\EB} problem in planar DAGs.

Another central problem in Algorithms and Complexity Theory is the problem of finding the perfect matching (denoted as {\PM}). The best known upper bound for {\PM} is \emph{non-uniform {\SPL}} \cite{ARZ99} and the best hardness known is {\NL}-hardness \cite{CSV84}. However, {\PM} in planar graph is known to be {\L}-hard \cite{DKLM10}. If we consider the planar bipartite graph, then {\PM} problem is known to be in {\UL} \cite{DGKT12}. {\PM} problem in bipartite graphs can be solved in polynomial time using Ford-Fulkerson algorithm for network flow \cite{KT05} but that takes the space linear in number of edges present in the graph. Unfortunately, no sublinear ($O(n^{1-\epsilon})$, for any $\epsilon >0$) space and polynomial time algorithm is known for {\PM} in planar bipartite graphs. Same is true for the problem of finding Hall-obstacle (denoted as {\HO} (Decision + Construction) for planar bipartite graphs, whereas it is known from \cite{DGKT12}, that {\HO} (Decision) is in {\co-UL} and {\HO} (Construction) is in {\NL}, when the graph under consideration is planar bipartite.\\
The problem {\ExPM} (Decision) (first posed in \cite{PY82}) denotes the problem of deciding the presence of a perfect matching in a given graph $G$ with edges coloured with Red or Blue, containing exactly $k$ Red edges for an integer $k$. This problem is not even known to be in {\P}. Now we consider a natural relaxation of the above problem just by concentrating on the perfect matching containing even number of Red edges and denote this problem as \emph{{\EPM}}. {\EPM} problem is in {\P} for bipartite graphs and in {\NL} for planar bipartite graphs \cite{DGKT12}. Till now, we do not know about any sublinear space and polynomial time algorithm for {\EPM} problem when concentrating only on planar bipartite graphs.

\subsection*{Our Contribution}
In this paper, we prove the following results.
\begin{theorem}
 \label{thm:shortpath}
 For directed planar graphs (containing no negative weight cycle and weights are bounded by polynomial in $n$) and for any constant $ 0 < \epsilon < \frac{1}{2} $, there is an algorithm that solves {\ShP} problem in polynomial time and $ O(n^{\frac{1}{2}+\epsilon}) $ space, where $n$ is the number of vertices of the given graph.
\end{theorem}
We use the space efficient construction of separator for planar graphs \cite{INPVW13}, and this is one of the main building blocks for the results stated in this paper. Let the separator be $S$. Now calculate the shortest distance of every $v \in S$ from the vertex $s$. The shortest path from $s$ to $t$ must pass through the vertices in the separator and thus knowing the shortest path from $s$ to each of such vertex is enough to find the shortest path from $s$ to $t$. The shortest path from $s$ to any $v \in S$ can be found by applying the same shortest path algorithm recursively to each of the connected component of the graph induced by $V \setminus S$. As a base case we use {\BellmanFord} algorithm to find the shortest path. The recursive invocation of the above technique will lead to the time-space bound mentioned in the above theorem.\\
Another main contribution of this paper is to give an algorithm for the {\RB} problem in planar DAG. The main idea behind our algorithm is to use a modified version of DFS algorithm along with the planar separator.
\begin{theorem}
 \label{thm:redbluepath}
 For any constant $ 0 < \epsilon < \frac{1}{2} $, there is an algorithm that solves {\RB} problem in planar DAG in polynomial time and $ O(n^{\frac{1}{2}+\epsilon}) $ space.
\end{theorem}
Now using the reduction given in \cite{Kul09} and the algorithm stated in the above theorem, we get an algorithm to solve the directed reachability problem for a fairly large class of graphs described in Section \ref{sec:redblue}, that takes polynomial time and $O(\sspace)$ space. Thus we can able to beat the bound given by Barnes, Buss, Ruzzo and Schieber \cite{BBRS92} for such class of graphs.\\
In this paper, we also establish a relation between {\EB} problem in a planar DAG and the problem of finding odd length cycle in a directed planar graph and thus we argue that both this problem has the same simultaneous time-space complexity. We use two colors Red and Blue to color the vertices of the given graph and then use the color assigned to the vertices of the separator to detect the odd length cycle. The conflicting assignment of color to the same vertex in the separator will lead to the presence of an odd length cycle. Here also we use the recursive approach to color the vertices and as a base case we use the well known {\tt BFS} algorithm to solve the problem of detecting odd length cycle in each small component. Thus we get the following theorem regarding solving {\EB} problem.
\begin{theorem}
 \label{thm:evenpath}
 For any constant $ 0 < \epsilon < \frac{1}{2} $, there is an algorithm that solves {\EB} problem in planar DAG in polynomial time and $ O(n^{\frac{1}{2}+\epsilon}) $ space.
\end{theorem}

Our another contribution is to give an time-space efficient algorithm for perfect matching problem in case of planar bipartite graphs.
\begin{theorem}
 \label{thm:perfectmatching}
 In Planar Bipartite Graphs, for any constant $ 0 < \epsilon < \frac{1}{2} $,\\
 (a) {\PM} (Decision + Construction) can be solved in polynomial time and $ O(n^{\frac{1}{2}+\epsilon}) $ space.\\
 (b) {\HO} (Decision + Construction) can be solved in polynomial time and $ O(n^{\frac{1}{2}+\epsilon}) $ space.
\end{theorem}
We build on the Miller and Noar's algorithm \cite{MN89} for perfect matching in planar bipartite graph. We show that this algorithm takes polynomial time and $ O(n^{\frac{1}{2}+\epsilon}) $ space as the only hard part of this algorithm is to find the shortest distance. We also argue that problem of finding Hall obstacle is directly associated with the problem of finding negative weight cycle and thus get same simultaneous time-space bound for this problem as of the problem of detecting negative weight cycle.\\
Next we show that the complexity of even perfect matching in planar bipartite graph is same as the perfect matching problem in planar bipartite graph and deciding the presence of odd length cycle in directed planar graph. Thus we get the following theorem for {\EPM} problem.
\begin{theorem}
 \label{thm:evenPM}
 For any constant $ 0 < \epsilon < \frac{1}{2} $, there exists an algorithm that solves {\EPM} in planar bipartite graphs in polynomial time and $ O(n^{\frac{1}{2}+\epsilon}) $ space.
\end{theorem}

The rest of the paper is organized as follows. In the next section, we give some notations and definitions used in this paper. In Section \ref{sec:shortestpath}, we give an algorithm for shortest path problem in directed planar graphs. In Section \ref{sec:redblue}, we give a simultaneous time-space bound for deciding the presence of Red-Blue Path in a planar DAG and then establish a relation between the problem of deciding the presence of an odd length cycle in directed planar graphs with the problem of deciding the presence of even path between two given vertices in planar DAG and thus give the same simultaneous time-space bound for both of these problems. And finally in Section \ref{sec:matching}, we discuss the simultaneous time-space bound of some matching problems in planar bipartite graphs.

\section{Preliminaries}
\label{sec:prelims}
\subsection{Notations}
A graph $ G=(V,E)$ consists of a set of vertices $V$ and a set of edges $ E $ where each edge can be represented as an ordered pair $(u,v)$ in case of directed graph and as an unordered pair $\{u,v\}$ in case of undirected graph, such that $ u,v \in V $. Unless it is specified, in this paper $G$ will denote the directed graph, where $|V|=n$. Given an graph $G$ and a set of vertices $X$, $G[X]$ denotes the subgraph of $G$ induced by $X$ and $ V(G) $ denotes the set of vertices present in the graph $ G $.

\subsection{Separator and Directed Planar Reachability}
The notions of separator and separator family defined below are crucial in this paper.
\begin{definition}
A subset of vertices $S$ of an undirected graph $G$ is said to be a $ \rho $-separator (for any constant $ \rho $, $ 0 < \rho < 1 $) if removal of $S$ disconnects $G$ into two sub-graphs $A$ and $B$ such that $ |A|,|B| \leq \rho n $ and the size of the separator is the number of vertices in $S$.

A subset of vertices $ \overline{S} $ of an undirected graph $G$ with $n$ vertices is said to be a $r(n)$-separator family if the removal of $ \overline{S} $ disconnects $G$ into sub-graphs containing at most $r(n)$ vertices.
\end{definition}

Now we restate the results and the main tools used in \cite{INPVW13} to solve directed planar reachability problem and these results are extensively used in this paper. In \cite{INPVW13}, the authors construct a $ \frac{8}{9} $-separator for a given undirected planar graph.

\begin{theorem}[\cite{INPVW13}]
\label{thm:separator}
(a) Given an undirected planar graph $G$ with $n$ vertices, there is an algorithm {\PlSep} that outputs a $ \frac{8}{9} $-separator of $G$ in polynomial time and $\tildeO (\sqrt{n}) $ space.\\
(b) For any $ 0 < \epsilon < 1/2 $, there is an algorithm {\PlSepFam} that takes an undirected planar graph as input and outputs a $ n^{1- \epsilon} $-separator family of size $ O(\ssize) $ in polynomial time and $ \tildeO(\sspace) $ space.
\end{theorem}
In \cite{INPVW13}, the above lemma is used to obtain a new algorithm for reachability in directed planar graph.
\begin{theorem}[\cite{INPVW13}]
\label{thm:planar-reach}
For any constant $0 < \epsilon < 1/2$, there is an algorithm {\algoPlanarReach} that, given a directed planar graph $G$ and two vertices $s$ and $t$, decides whether there is a path from $s$ to $t$. This algorithm runs in time $n^{O(1/\epsilon)}$ and uses $O(n^{1/2+\epsilon})$ space, where $n$ is the number of vertices of $G$.
\end{theorem}

\section{Shortest Path Problem in Directed Planar Graph}
\label{sec:shortestpath}
Let {\ShP} be the problem of computing the shortest distance and the corresponding path between two vertices in a graph $G$. For a given graph $G=(V,E)$ with a weight function $w:E \rightarrow \mathbb{R}$ (negative weights are also allowed), and two vertices $s$ and $t$, let $ dist^w_{G}(s,t) $ denote the shortest distance and $path^w_{G}(s,t)$ denote the shortest path from $s$ to $t$. We will consider the weight assigned to an edge is bounded by some polynomial in $n$, where $n$ is the number of vertices in $G$.
Note that single-source shortest path problem and all-pair shortest path problem can be solved by executing the algorithm used for determining shortest path from $s$ to $t$, polynomially many times.

As a consequence of the above result, we can also detect negative weight cycle in a given directed planar graph. If we determine the shortest path from $s$ to all other vertices then negative cycle must lie in any one of these paths and shortest distance of that path is negative infinity. Now consider the {\ShP}  problem when the given graph $G$ is directed planar. We will also briefly discuss the procedure of detecting negative weight cycle for directed planar graphs, in the later part of this section.

\begin{proof}[Proof of Theorem \ref{thm:shortpath}]
Let $ G=(V,E) $ be the given directed planar graph where $|V|=n$. Consider any constant $0<\epsilon <1/2$. Let $ \overline{S} $ be the $ n^{(1-\epsilon)}  $-separator family computed by {\PlSepFam} on underlying undirected graph of $ G $ and $ S= \overline{S} \cup \{s,t\} $. Define an array $ C_{s} $ of size $|S|$. $ C_{s}[i] $ will store the distance of $i$-th vertex $v_i$ of the set $ S $ from the vertex $s$ and initially $ C_{s}[s] =0$, $ \forall _{v_i\ne s},  C_{s}[i] =\infty $. For calculating the shortest distance, we use {\BellmanFord} if the graph is small; otherwise, we calculate the shortest distance of the vertices in $S$ from $s$ through each connected component of $G[V\setminus S]$ and then use those distances to calculate the shortest distance from $s$ to $t$ in the overall graph. We do this recursively to achieve the desired time and space bound. To find the shortest distance between $s$ 
and $t$ in the 
given graph $G$, we run {\PlShortPath} (Algorithm \ref{algo:PlanarShortestPath}) with the input $(G,s,t,n,S,C_s)$, where $n$ is the number of vertices of $G$.\\
In Algorithm \ref{algo:PlanarShortestPath}, within the loop [9-15], we evaluate the shortest distance of all the vertices $v_i \in S$ from $s$ through each connected component of $G[V\setminus S]$ and by update $C_s$ (in line 13), we mean that update the entry in $C_s[i]$ for some $i$, if the currently calculated distance of $v_i$ from $s$ is smaller than the previously stored one. We run the loop [9-15] total $|S|$ number of times and the reason for that is mentioned during the correctness proof.

\begin{algorithm}[h]
\SetAlgoLined
  \Input{$ G'=(V',E'),s',t',n,T,A $}
 \Output{$dist^w_{G'}(s',t')$}
 //let $ r'=n'^{(1-\epsilon)} $, $|V'|=n'$\\
  \eIf{$ n' \le n^{\frac{1}{2}} $}{
   {Run \BellmanFord}($ G',s',t' $) using the values stored in $A$ and return the shortest distance\;
   }{
   
   Run {\PlSepFam} on the underlying undirected graph of $G'$ to compute $ r' $-separator family $ \overline{S'} $\;
   Set $ S' : = \overline{S'} \cup \{s',t'\} $\;
   Define an array $ C_{s'} $ of size $|S'|$. $ C_{s'}[i] $ will store $dist^w_{G'}(s',v_i)$, where $v_i$ is the $i$-th vertex of the set $ S' $ and set $ \forall _{v_i \in T}, C_{s'}[i]=A[i],  \forall _{v_i \not \in T}, C_{s'}[i]=\infty $\\
   \For{$ round=1 \: to \: |S'|  $}
   {
   \For{every $ x \in V' $}
    {
   	\For{every $ v \in S' $}
   	{
   	//$ V_x= $ the set of vertices of undirected version of $ G[V'\setminus S'] $'s connected component containing $x$.\\
   		Run {\PlShortPath}($ G[V_x \cup S'],s',v,n,S',C_{s'} $)\;
   		Update $ C_{s'} $\;
   	}
    }
   }
  }
   \caption{Algorithm {\PlShortPath}: Shortest Distance in Directed Planar Graph}
   \label{algo:PlanarShortestPath}
\end{algorithm}

We use {\PlShortPath} as a subroutine to report the shortest path. The algorithm is stated as {\PlShortReport} (Algorithm \ref{algo:ReportPlanarShortestPath}) and to report the shortest path between $s$ and $t$ in $G$, we run this algorithm with the input $(G,s,t,n,S,C_s)$.

\begin{algorithm}[h]
\SetAlgoLined
\Input{$ G'=(V',E'),s',t',n,T,A $}
\Output{$path^w_{G'}(s',t')$ in reverse order}
\BlankLine

Run {\PlShortPath}($ G',s',t',n,T,A $)\;
//let $ \overline{S'} $ be the corresponding separator family and $ S'=\overline{S'} \cup \{s',t'\} $ and $ C_{s'} $ be the corresponding array storing the shortest distances of vertices in separator family from $ s' $\\
//let $ N(t') $ be the set of neighbor vertices of the vertex $ t' $\\
Define an array $ C_{t'} $ of size $|S'|$. Initialize $ C_{t'}[t'] =0$, $ \forall _{v_i\ne s'},  C_{s}[i] =\infty $.\\
  \For{every $ \: x \in N(t') $}{
  	\For{every $ v \in \overline{S'} $}{
  	     //let $  G_{rev} $ be the graph with the same set of vertices as $G$ but the direction of the edges are reversed\\
   Run {\PlShortPath}($ G_{rev}[V_x \cup S'],t',v,n,S',C_{t'} $)\;
   }
  }
  Let $ v' \in \overline{S'} $ such that $dist^w_{G_{rev}[V_{x'} \cup S']}(t',v')+dist^w_{G'}(s',v')=dist^w_{G'}(s',t')$\\
    \eIf{$ |V_{x'} \cup S'| \le n^{\frac{1}{2}} $}{
     {Run \BellmanFord}($ G_{rev}[V_{x'} \cup S'],t',v' $) and report $path^w_{G_{rev}[V_{x'} \cup S']}(t',v')$\;
     }{
     Run {\PlShortReport}($ G_{rev}[V_{x'} \cup S'],t',v',n,S',C_{t'} $)\;
  }
  Reinitialize $A$ and Run {\PlShortReport}($ G',s',v',n,T,A $)\;

 \caption{Algorithm {\PlShortReport}: Report Shortest Path in Directed Planar Graph}
   \label{algo:ReportPlanarShortestPath}
\end{algorithm}

Let $ \mathcal{S} $ and $ \mathcal{T} $ denote its space and time complexity functions. Since $ (1-\epsilon)^k \le \frac{1}{2} $ for $ k=O(\frac{1}{\epsilon}) $, the depth of the recursion is $ O(\frac{1}{\epsilon}) $. Also, $ |V_x \cup S'| \le 2 n'^{(1-\epsilon)} $. This gives us the following recurrence relation:
\[ 
\mathcal{S}(n') = 
\begin{cases}
\text{\~{O}}(n'^{(\frac{1}{2}+\epsilon) })+\mathcal{S}(2n'^{(1-\epsilon)}) & \text{if } n' > n^{\frac{1}{2}}\\
\text{\~{O}}(n^{\frac{1}{2}}) & \text{otherwise}
\end{cases}
\]
Thus, $ \mathcal{S}(n) = O(\frac{1}{\epsilon}) \text{\~{O}}(\ssize) = \text{\~{O}}(\ssize) $.\\
For time analysis, we get the following recurrence relation:
\[
\mathcal{T}(n') = 
\begin{cases}
q(n) (p_1 (n') \mathcal{T}(2 n'^{(1-\epsilon)}) + p_2(n') ) & \text{if } n' > n^{\frac{1}{2}}\\
q(n) \text{\~{O}}(n^{\frac{1}{2}}) & \text{otherwise}
\end{cases}
\]
As the recursion depth is bounded by $ O(\frac{1}{\epsilon}) $ (a constant) and the subroutine {\PlShortPath} is called polynomial many times, we have $ \mathcal{T}(n)=p(n) $ for some polynomial $ p(n) $. Using an extension of the above idea, it can easily be seen that {\PlShortReport} also uses $ \tildeO(\sspace) $ space and polynomial time.

We next show using induction the correctness of the above algorithms. Let $ G'=(V',E'),s',t',n,S,C_s $ be an instance of {\PlShortPath}. When $ n' \le n^{\frac{1}{2}} $, the correct answer is given since it is just the execution of the {\BellmanFord} algorithm.
Now consider the shortest path $P$ from $s$ to $t$, which can be decomposed as $ s \xrightarrow{P_1} v_1 \xrightarrow{P_2} v_2 \cdots v_k \xrightarrow{P_k} t $, where each $P_i$, for $i\ne 1,k$, is the shortest path between $v_{i-1}$ and $v_i$ through a $ O(n'^{(1-\epsilon)}) $ sized connected region. Now by induction on $n$, we can say that $ s \xrightarrow{P_1} v_1 \xrightarrow{P_2} v_2 \cdots v_{i-1} \xrightarrow{P_i} v_i $ is the shortest 
path from $s$ to $v_i$. As the size of separator family is $ \tildeO(\ssize) $ and each path going from one $ O(n'^{(1-\epsilon)}) $ sized region to other must pass through a vertex in the separator family $\overline{S}$, so $ k \le |\overline{S}| $ and thus the execution of loop [9-15] in Algorithm \ref{algo:PlanarShortestPath} total $|S|$ number of times suffice to output the shortest distance between $s$ and $t$. The above argument is similar to the proof of correctness of {\BellmanFord} algorithm. Now as $ |\overline{S}| \le \tildeO(\ssize) $, the length of the shortest path from $s$ to $t$ is at most $ \tildeO(\ssize)O(n'^{(1-\epsilon)}) $.

Using the similar argument, it is not hard to see that {\PlShortReport} will correctly output the shortest path from $s$ to $t$ in the reverse order.
\end{proof}

Now if there exists a negative weight cycle in the given graph, then either that will be completely inside $ O(n^{(1-\epsilon)}) $ sized region or it must pass through at least two vertices of the separator family. To detect the negative weight cycle we use a slight modification of {\PlShortPath} algorithm and in the modified version we will run the {\PlShortPath} slightly more than $ |S| $ times and if in the last run any of the value of $ C_s[v] $ changes then we can infer that there is a negative weight cycle and which lies in the path of $s$ to $v$ and we can report the negative weight cycle just using the {\PlShortReport} algorithm. Thus the detecting and reporting negative weight cycle problem also have the same space and time (up to polynomial blow up) complexity as that of shortest path problem.
\begin{corollary}
 \label{cor:negcycle}
 For directed planar graphs, for any constant $ 0 < \epsilon < 1/2 $, there is an algorithm that solves the problem of detecting negative weight cycle in polynomial time and uses $ O(n^{\frac{1}{2}+\epsilon}) $ space, where n is the number of vertices of the given graph $G$.
\end{corollary}

\section{Red-Blue Path Problem}
\label{sec:redblue}
\subsection{Deciding Red-Blue Path in Planar DAG}
\label{subsec:redbluepath}
Given a directed graph $G$ with each edge colored with Red or Blue and two vertices $s$ and $t$, the {\RB} problem decides whether there exists a directed path from $s$ to $t$ that alternate between Red and Blue edges such that the first edge is Red and last one is Blue. The {\RB}  problem is a generalization of the reachability problem in graphs, however this problem is {\NL}-complete even when restricted to planar DAGs \cite{Kul09}. This makes it an interesting problem in the area of space bounded complexity as to the best of our knowledge, this is the only ``reachabililty-like'' problem in planar graphs that is hard for {\NL}. We will now give a proof of Theorem \ref{thm:redbluepath}, stated in the Introduction.

\begin{proof}[Proof of Theorem \ref{thm:redbluepath}]
Consider a planar DAG $G$. Let $ \overline{S} $ be the $ n^{(1-\epsilon)}  $-separator family computed by {\PlSepFam} on underlying undirected graph of $ G=(V,E) $ and $ S= \overline{S} \cup \{s,t\} $. Now we devise a process on the vertices of $S$, which is similar to normal DFS algorithm. For the sake of convenience, we associate two numerical values to the edge colors -- $0$ to Red and $1$ to Blue. We run {\tt RedBluePath} with the input $( G,s,t,n,0,1 )$ and if the returned value is true, then we say that there is a directed path from $s$ to $t$ such that the two consecutive edges are of different colors and the first edge is Red and last one is Blue. In Algorithm \ref{algo:ModifiedColoredDFS}, we use the notation $(u,v)\in^{(init,temp)} \overline{E'}$ to decide whether there is a path from $u$ to $v$ that starts with an edge of color value $init$ and ends with an edge of color value $temp$, alternating between Red and Blue edges.

\begin{algorithm}
 \SetAlgoLined
\Input{$ G'=(V',E'),s',t',init,final $}
\Output{``Yes'' if there is a valid red blue alternate path from $s'$ and $t'$ starts with $init$ and ends with $final$}
   //Use two sets- $N_i$, for $i=0,1$, to store all the vertices that have been explored with the color value $i$\\
    \If{$s' \not \in N_{init} $}{
    Add $s'$ in $N_{init}$\;
     \For{each edge $(s',v)\in E'$ of color value $init$}{
	\If{$v = t'$ and $init=final$}{
	  Return true\;
	}
	Run {\tt ColoredDFS}($ G',v,t',init+1(mod$ $2),final $)
      }
    }
   \caption{Algorithm {\tt ColoredDFS}: One of the Building Blocks of {\tt RedBluePath}}
\end{algorithm}

\begin{algorithm}
 \SetAlgoLined
 \Input{$ \overline{G'}=(\overline{V}',\overline{E'}),G',s',t',init,final $}
 \Output{``Yes'' if there is a valid red blue alternate path from $s'$ and $t'$ starts with $init$ and ends with $final$}
    
    //Use two sets- $R_i$, for $i=0,1$, to store all the vertices that have been explored with the color value $i$\\
   \If{$s' \not \in R_{init} $}{
    Add $s'$ in $R_{init}$\;
     \For{each $(s',v)\in^{(init,temp)} \overline{E'}$ for each $temp\in\{0,1\}$}{
	\If{$v = t'$ and $temp=final$}{
	  Return true\;
	}
	Run {\tt ModifiedColoredDFS}($ \overline{G}',G',v,t',temp+1 $ (mod $2$), $final$)
      }
    }
    //``$(u,v)\in^{(init,temp)} \overline{E'}$?'' query will be solved using the following procedure\\
    \For{every $ a \in V $}{
      //$V$ be the set of vertices of $G'$\\
      //$V_a=$ the set of vertices of undirected version of $ G[V\setminus \overline{V}'] $'s connected component containing $a$\\
      \If{{\tt RedBluePath}($G[V\setminus \overline{V}'],u,v,n,init,temp$) is true}{
	Return true for the query\;
      }
    }
    Return false for the query\;
    //End of the query procedure
   \caption{Algorithm {\tt ModifiedColoredDFS}: One of the Building Blocks of {\tt RedBluePath}}
   \label{algo:ModifiedColoredDFS}
\end{algorithm}

\begin{algorithm}
 \SetAlgoLined
\Input{$ G',s',t',n,init, final $}
\Output{``Yes'' if there is a valid red blue alternate path from $s'$ and $t'$ starts with $init$ and ends with $final$}
  \eIf{$ n' \le n^{\frac{1}{2}} $}{
   Run {\tt ColoredDFS}($ G',s',t',init, final $)\;
   }{
   //let $ r'=n'^{(1-\epsilon)} $\\
   Run {\PlSepFam} on the underlying undirected graph of $G'$ to compute $ r' $-separator family $ \overline{S'} $\;
   Run {\tt ModifiedColoredDFS}($ \overline{G'}=(\overline{S'}\cup \{s',t'\},\overline{E'}),G',s',t',init,final $)\;
  }
   
   \caption{Algorithm {\tt RedBluePath}: Algorithm for Red-Blue Path in planar DAG}
\end{algorithm}

The base case takes $\tildeO(n^{1/2})$ space and polynomial time. The sets $N_0$ and $N_1$ of algorithm {\tt ColoredDFS} only store all the vertices of the input graph and we run {\tt ColoredDFS} on a graph with $n^{1/2}$ many vertices and it visits all the edges of the input graph at most once which results in the polynomial time requirement. Now by doing the same analysis as that of {\PlShortPath}, it can be shown that {\tt RedBluePath} will take $ O(\sspace)$ space and polynomial time.\\
We now give a brief idea about the correctness of this algorithm. In the base case, we use similar technique as DFS just by alternatively exploring Red and Blue edges and thus this process gives us a path where two consecutive edges are of different colors. Otherwise, we also do a DFS like search by alternatively viewing Red and Blue edges and we do this search on the graph $H=(\overline{S'}\cup \{s,t\},\overline{E'})$. By this process, we decide on presence of a path in $H$ from $s$ to $t$ such that two consecutive edges are of different colors in $G$ and the edge coming out from $s$ is Red and the edge coming in at $t$ is Blue. This is enough as each path $P$ in $G$ must be broken down into the parts $P_1,P_2,\cdots ,P_k$ and each $P_i$ must be a sequence of edges that starts and ends at some vertices of $ \overline{S'}\cup \{s,t\} $ and also  alternates in color. We find each such $P_i$, just by considering each connected component of $ G(V' \setminus \overline{S'}) $ and repeating the same steps recursively. 
\end{proof}

Due to \cite{Kul09}, we know that the reachability problem in directed graphs reduces to {\RB} in planar DAG. The class of graphs for which this reduction results in increase the number of vertices less than a quadratic factor of the number of vertices of the original graph, we have an algorithm for reachability problem that takes sublinear space and polynomial time. As a special case of this we can state the following theorem.
\begin{theorem}
  \label{thm:sparsereach}
 Given a directed acyclic graph $G=(V,E)$, where $|E|=\tildeO(n)$, with a drawing in a plane such that the number of edge crossings is $\tildeO(n)$ and two vertices $s$ and $t$, then for any constant $0<\epsilon<\frac{1}{2}$, there is an algorithm that decides whether there is a path from $s$ to $t$ or not. This algorithm runs in polynomial time and uses $O(\sspace)$ space, where $n$ is the number of vertices of $G$.
\end{theorem}
\begin{proof}
 We consider a reduction similar to the reduction from directed reachability problem to {\RB} problem in planar DAG given in \cite{Kul09}. We do the following: (i) insert new vertices in between edges of $G$ so that in the resulting graph each edge takes part in only one crossing and (ii) replace each crossing of the resulting graph with a \emph{planarizing gadget} as in Fig. \ref{fig:gadget} and also replace each edge without any crossing with two edges as shown in Fig. \ref{fig:gadget}. Denote the resulting graph as $G_{planar}$ and the corresponding vertices of $s$ and $t$ as $s'$ and $t'$. It is easy to see that there is a bijection between $s-t$ paths in $G$ and $s'-t'$ paths in $G_{planar}$ that starting with Red edge alternates between Red and Blue edges and finally ends with Blue edge.\\
 \begin{figure}[htb]
 \begin{center}
  \includegraphics{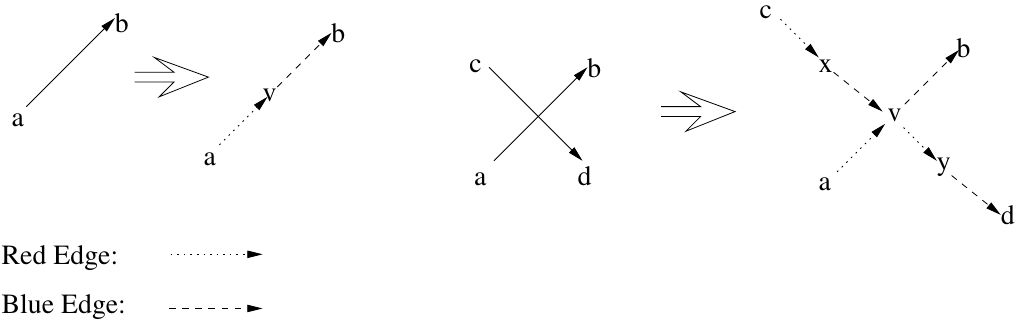}
  \caption{Red-Blue Edge Gadget}
  \label{fig:gadget}
 \end{center}
\end{figure}
 Now we count the number of vertices in $G_{planar}$. If the drawing of the given graph $G$ contains $k$ edge crossings, then the step (i) will introduce at most $2k$ many vertices and say after this step the number of edges becomes $ m $. Then step (ii) will introduce at most $ (2m+3k) $ many vertices. Thus the graph $G_{planar}$ contains $\tildeO(n)$ many vertices and then by applying the algorithm {\tt RedBluePath} on $G_{planar}$, we get the desired result.

\end{proof}

A large class of graphs will satisfy the conditions specified in Theorem \ref{thm:sparsereach}. We now explicitly give an example of one such class of graphs. Before that, we give some definitions. \emph{Crossing number} of a graph $G$, denoted as $cr(G)$, is the lowest number of edge crossings (or the crossing point of two edges) of a drawing of the graph $G$ in a plane. A graph is said to be \emph{$k$-planar} if it can be drawn on the plane in such a way that each edge has at most $k$ crossing point (where it crosses a single edge). It is known from \cite{PT97} that a $k$-planar graph with $n$ vertices has at most $O(n \sqrt{k})$ many edges. It is easy to see that a $k$-planar graph has crossing number at most $mk$, where $m$ is the number of edges. Now we can state the following corollary.
\begin{corollary}
 \label{cor:kplanarDAG}
 Given a directed acyclic graph, which is $k$-planar, where $k=O(\log^c n)$, for some constant $c$, with a drawing in a plane having minimum number of edge crossings and two vertices $s$ and $t$, then for any constant $0<\epsilon<\frac{1}{2}$, there is an algorithm that decides whether there is a path from $s$ to $t$ or not. This algorithm runs in polynomial time and uses $O(\sspace)$ space, where $n$ is the number of vertices of the given graph.
\end{corollary}

\subsection{Deciding Even Path in Planar DAG}
\label{subsec:evenpath}
In a given directed graph $G$ with each edge colored with Red or Blue, {\EB} is the problem of deciding the presence of a (simple) path between two given vertices $s$ and $t$, that contains even number of edges. We can view this problem as a relaxation of {\RB} problem as a path starting with Red edge and ending with Blue edge is always of even length. In this section, we establish a relation between {\EB} problem in planar DAG with detecting a odd length cycle in a directed planar graph with weight one (can also be viewed as an unweighted graph).

\begin{proposition}
  \label{thm:oddccycle}
  For directed planar graphs, for any constant $ 0 < \epsilon < \frac{1}{2} $, there is an algorithm that solves the problem of deciding the presence of odd length cycle in polynomial time and $ O(n^{\frac{1}{2}+\epsilon}) $ space, where $n$ is the number of vertices of the given graph.
\end{proposition}
The above proposition is true due to the fact that we can do BFS efficiently for undirected planar graph and it is enough to detect odd length cycle in each of the strong components of the undirected version of the given directed planar graph.
For undirected graph, presence of odd length cycle can be detected using BFS algorithm and then put red and blue colors on the vertices such that vertices in the consecutive levels get the opposite colors. After coloring of vertices if there exists an monochromatic edge (edge where both vertices get the same color), then we can conclude that there is an odd length cycle in the graph otherwise there is no odd length cycle. But this is not the case for general directed graph. However, the following observation will help us to detect odd length cycle in directed graph. In the following proof, we use $u \rightarrow v$ to denote a directed edge $(u,v)$ and $x \xrightarrow{P} y$ to denote a directed path $P$ from a vertex $x$ to $y$.
\begin{observation}
\label{lem:oddundirected}
A strongly connected directed graph contains an odd length cycle if and only if its undirected version contains an odd length cycle.
\end{observation}
\begin{proof}
If a strongly connected directed graph contains an odd length cycle then its undirected version contains an odd length cycle and so there is nothing to prove.\\
Now to prove the reverse direction, we will use the induction arguments on the length of the odd cycle in the undirected version of the graph. The base case is when the undirected version of the graph contains a $3$-length cycle. If the undirected edges present in the undirected cycle also form directed cycle when we consider the corresponding edges in the directed graph, then there is nothing to prove. But if this is not the case, then the Fig. \ref{fig:oddlenbase} will depict the possible scenarios. As the graph is strongly connected, so there must be a path $P$ from $t$ to $s$ and if this path does not pass through the vertex $x$, then any one of the following two cycles $ s \rightarrow t \xrightarrow{P} s $ or $ s \rightarrow x \rightarrow t \xrightarrow{P} s $ must be of odd length. Now suppose $P$ contains the vertex $x$ and thus $ P = P_1 P_2 $, where $ P_1 $ is the path from $t$ to $x$ and $ P_2 $ is the path from $x$ to $s$. It is easy to see that all the three cycles $ s \rightarrow t \xrightarrow{P} s $, 
$ x \rightarrow t \xrightarrow{P_1} x $ and $ s \rightarrow x \xrightarrow{P_2} s $ cannot be of even length.\\

\begin{figure}[htb]
 \begin{center}
  \includegraphics{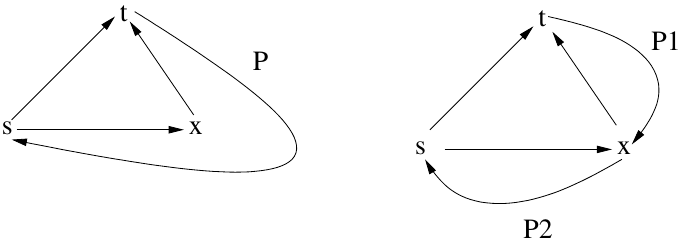}
  \caption{For undirected cycle of length $3$}
  \label{fig:oddlenbase}
 \end{center}
\end{figure}

Now by induction hypothesis, assume that if the undirected version has a cycle of $k$-length ($k$ odd), then there exists an odd length cycle in the original directed graph.\\

\begin{figure}
 \begin{center}
  \includegraphics{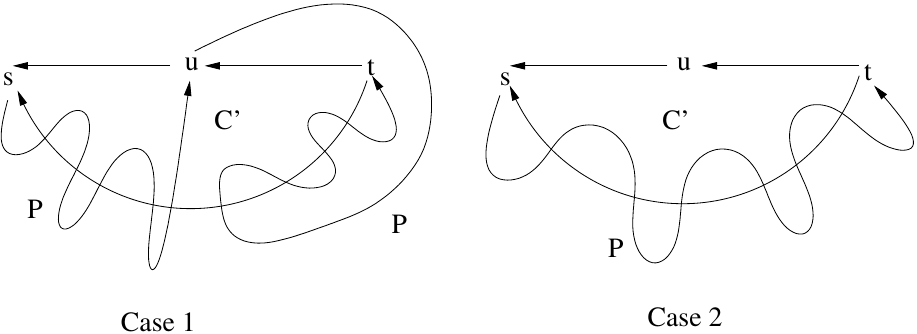}
  \caption{For undirected cycle of length $(k+2)$}
  \label{fig:oddleninduction}
 \end{center}

\end{figure}

Now lets prove this induction hypothesis for any undirected cycle of length $(k+2)$. Consider the corresponding edges in the directed graph and without loss of generality assume that this is not a directed cycle. As $(k+2)$ is odd, so there must be one position at which two consecutive edges are in the same direction. Now contract these two edges in both directed and undirected version of the graph and consider the resulting $k$-length cycle in the undirected graph. So according to the induction hypothesis, there must be one odd length cycle $C$ in the resulting directed graph. Now if $C$ does not contain the vertex $u$ (where we contract the two edges), then expanding the contracted edges will not destroy that cycle and we get our desired odd length cycle in the directed version of the graph. But if this is not the case, then consider $C$ after expanding those two contracted edges($ t \rightarrow u \rightarrow s $), say the resulting portion is $ C' $. If $C'$ is a cycle, then there is nothing more to do. 
But if not, then consider the path $P$ from $s$ to $t$ (there must be such path as the graph is strongly connected). Now  there will be two possible cases: either $P$ contains $ u $ or not. It is easy to see that for both the possible cases (case $1$ and case $2$ of Fig. \ref{fig:oddleninduction} and in that figure every crossing of two paths denotes a vertex), all cycles generated by $ C' $ and $P$ cannot be of even length. In case $1$, if all the cycles generated by the paths $ s \xrightarrow{P} u $ and $ t \xrightarrow{C'} s $ and all the cycles generated by the paths $ u \xrightarrow{P} t $ and $ t \xrightarrow{C'} s $ are of even length, then as $ t \xrightarrow{C'} s $ is of odd length, so the path $ s \xrightarrow{P} u \xrightarrow{P} t $ must be of odd length. And then one of the following two cycles  $ s \xrightarrow{P} u \rightarrow s $ and $ u \xrightarrow{P} t \rightarrow u $ is of odd length. Similarly in case $2$, if all 
the cycles generated by $ s \xrightarrow{P} t $ and $ t \xrightarrow{C'} s $ are of odd length, then the path $ s \xrightarrow{P} t $ is 
of odd length and so the cycle $ s \xrightarrow{P} t \rightarrow u \rightarrow s $ is of odd length.
\end{proof}

\begin{proof}[Proof of Proposition \ref{thm:oddccycle}]
In a directed graph, any cycle cannot be part of two different strong component, so checking presence of odd cycle is same as checking presence of odd cycle in each of its strong components. Constructing strong components of a directed planar graph can be done by polynomial many times execution of {\algoPlanarReach} algorithm, as a strong component will contain vertices $x$, $y$ if and only if {\algoPlanarReach} ($G,x,y,n$) and {\algoPlanarReach} ($G,y,x,n$) both return ``yes''. And thus strong component construction step will take $ \tildeO(\sspace) $ space and polynomial time. After constructing strong components, it is enough to check presence of odd cycle in its undirected version (according to Observation \ref{lem:oddundirected}). So now on, without loss of generality, we can assume that the given graph $G$ is strongly connected and let $G_{undirec}$ be the undirected version of $G$. Now execute {\PlOddCycle} ($G_{undirec},s,n $) (Algorithm \ref{algo:oddcycle}) after setting the color of $s$ (any arbitrary vertex) to red.\\
\begin{algorithm}[h]
 \SetAlgoLined
\Input{$ G'=(V',E'),s',n $, where $G'$ is an undirected graph}
\Output{``Yes'' if there is an odd length cycle}
  \eIf{$ n' \le n^{\frac{1}{2}} $}{
   Run {\tt BFS}($G'$ , $s' $) and color the vertices with red and blue such that vertices in the alternate layer get the different color starting with a vertex that is already colored\;
   \If{there is a conflict between stored color of a vertex and the new color of that vertex or there is an edge between same colored vertices}{
      		return ``yes''\;
      	}
   }{
   //let $ r'=n'^{(1-\epsilon)} $\\
   Run {\PlSepFam} on $G'$ to compute $ r' $-separator family $ \overline{S'} $\;
   Set $ S' : = \overline{S'} \cup \{s'\} $\;
   \For{every $ x \in V' $}
   {
   	//$ V_x=$ the set of vertices of $ G[V'\setminus S'] $'s connected component containing $x$.\\
   		Run {\PlOddCycle}($ G[V_x \cup S'],s',n $)\;
   		Store color of the vertices of $ S' $ in an array of size $|S'|$\;
   }
  }
   \caption{Algorithm {\PlOddCycle}: Checking Presence of Odd Cycle in an Undirected Planar Graph}
   \label{algo:oddcycle}
\end{algorithm}

By doing the similar type of analysis as that of {\PlShortPath}, it can be shown that {\PlOddCycle} will take $ O(\sspace)$ space and polynomial time and so over all space complexity of detecting odd length cycle in directed planar graph is $ O(\sspace)$ and time complexity is polynomial in $n$.\\
Now we argue on the correctness of the algorithm {\PlOddCycle}. This algorithm will return ``yes'' in two cases. First case when there is a odd length cycle completely inside a small region ($ n' \le n^{\frac{1}{2}} $) and so there is nothing to prove for this case as it is an well known application of BFS algorithm. Now in the second case, a vertex $v$ in the separator family will get two conflicting colors means that there exists at least one vertex $u$ in the separator family such that there is two vertex disjoint odd as well as even length path from $u$ to $v$ and as a result, both of these paths together will form an odd length cycle.
\end{proof}
Now we are ready to prove the main theorem of this subsection.

\begin{proof}[Proof of Theorem \ref{thm:evenpath}]
Given a planar DAG $ G $ and two vertices $s$ and $t$, first report the shortest path from $s$ to $t$, say $ P $, which we can do in polynomial time and $ O(n^{\frac{1}{2}+\epsilon}) $ space (by Theorem \ref{thm:shortpath}) and if this path is not of even length, then construct a directed graph $ G' $ which has the same vertices and edges as $ G $ except the edges in path $ P $, instead we do the following: if there is an edge $(u,v)$ in $ P $, then we add an edge $(v,u)$ in $ G' $. Now we can observe that the new graph $ G' $ is a directed planar graph. Now we claim the following.

\begin{claim}
 $ G $ has an even length path if and only if $ G' $ has an odd length cycle.
\end{claim}
\begin{proof}
Suppose $ G' $ has an odd length cycle, then that cycle must contains the reverse edges of $ P $ in $ G $. Denote the reverse of the path $ P $ by $ P_{rev} $. Now lets assume that the odd cycle $ C' $ contains a portion of $ P_{rev} $ (See Fig. \ref{fig:oddlencycle}). Assume that the cycle $ C' $ enters into $ P_{rev} $ at $x$ (can be $t$) and leaves $ P_{rev} $ at $y$ (can be $s$). Then in the original graph $ G $, the path $ s \xrightarrow{P} y \xrightarrow{C'} x \xrightarrow{P} t $ is of even length.\\
\begin{figure}[htb]
 \begin{center}
  \includegraphics{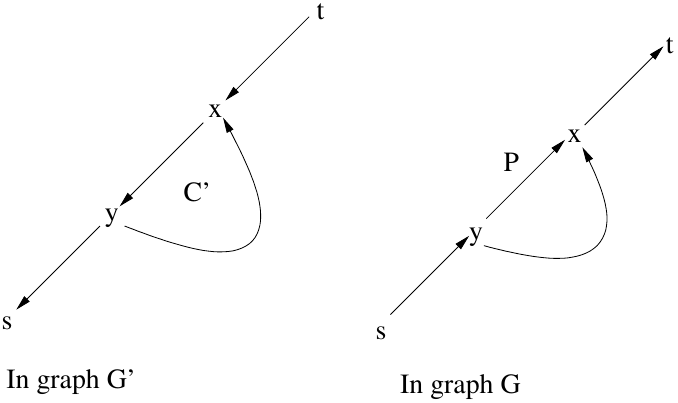}
  \caption{When $G'$ contains an odd length cycle}
  \label{fig:oddlencycle}
 \end{center}
\end{figure}
\begin{figure}[htb]
 \begin{center}
  \includegraphics{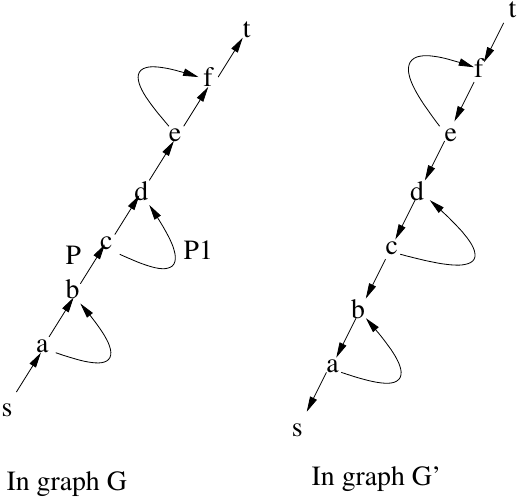}
  \caption{When $G$ contains an even length $s-t$ path}
  \label{fig:evenlenpath}
 \end{center}
\end{figure}
Now for the converse, lets assume that there exists an even length path $ P_{1} $ from $s$ to $t$ in $ G $. Both the paths $ P $ and $ P_{1} $ may or may not share some edges and without loss of generality we can assume that they share some edges (See Fig. \ref{fig:evenlenpath}). Now if we consider all the cycles formed by $ P_{rev} $ and portions of $ P_{1} $ in $ G' $, then it is easy to see that all the cycles cannot be of even length until length of $ P $ and $ P_{1} $ both are of same parity (either both odd or both even), but this is not the case.
\end{proof}
Now we can check the presence of odd length cycle in the graph $ G' $ in polynomial time and $ O(n^{\frac{1}{2} + \epsilon}) $ space (by Proposition \ref{thm:oddccycle}).
\end{proof}

\section{Perfect Matching in Planar Bipartite Graphs}
\label{sec:matching}
\subsection{Finding a Perfect Matching}
In a graph $G$, a \emph{matching} is a set of vertex disjoint edges where the end-points of these edges are \emph{matched}. A \emph{perfect matching} is a matching where every vertex is matched. In this section, we consider the following two matching problems.\\
(i) {\PM} (Decision): given a graph $G$, decide whether $G$ contains a perfect matching and\\(ii) {\PM} (Construction): given a graph $G$, construct a perfect matching (if exists).\\
We first discuss the Miller and Naor's algorithms ({\tt MN-Pseudo-Flow} \& {\tt MN-Decision}) for solving the decision version of the perfect matching problem in planar bipartite (undirected) graphs and also the algorithm to construct the perfect matching ({\tt MN-Construction}). Our main observation is that all the algorithms can be implemented in polynomial time and using $ O(n^{\frac{1}{2}+\epsilon}) $ space.\\
Before discussing the algorithms, we first define some terminology which will be used later.\\
A \emph{capacity-demand graph} of an undirected graph $G=(V,E)$ is defined as a triple $(G'=(V',E'),c,d)$, where $E'=\{(u,v)\mid\{u,v\}\in E\}$  and every edge $(u,v) \in E' $ is assigned a real valued \emph{capacity} $c(u,v)$ and every vertex $ v \in V' $ is assigned a real valued \emph{demand} $d(v)$.\\
A \emph{pseudo-flow} in a capacity-demand graph $(G=(V,E),c,d)$ is defined as a function $f:E  \rightarrow \mathbb{R} $ such that:\\
(i) for every edge $(u,v)$ $ \in E $, $ f(u,v)=-f(v,u) $ and \\
(ii) for every vertex $ v \in V $, $ \sum\limits_{w \in V:(v,w)\in E} f(v,w) = d(v) $\\
A \emph{flow} in a capacity-demand graph $(G=(V,E),c,d)$ is defined as a function $f : E \rightarrow \mathbb{R} $ such that:\\
(a) $f$ is a pseudo-flow in $(G,c,d)$ and\\
(b) for every $(u,v)$ $ \in E $, $ f(u,v) \le c(u,v) $\\
A \emph{zero-demand} graph $(G,c)$ is a capacity-demand graph where $d(v)=0, \forall_{v \in V}$.\\

\begin{definition}[Directed Dual]
Suppose the dual of a undirected planar graph $G=(V,E)$ respect to a fixed embedding is denoted by $ G^d=(V^d,E^d)$. Then the directed dual of $G$ is a directed graph denoted by $G^*=(V^*,E^*)$ such that $E^*=\{(u,v)\mid\{u,v\}\in E^d\}$.
\end{definition}
Now we are ready to mention the main lemma from \cite{MN89}.
\comment{
\begin{proposition}[\cite{MN89}]
Let $f$ be a flow in a zero-demand graph $(G,c)$. If $ C^* = (e_1^*,\cdots,e_k^*) $ is a directed cycle in $ \overleftrightarrow{G} $, then $ \sum\limits_{e:e^* \in C^*} f(e) = 0 $
\end{proposition}}
\begin{lemma}[\cite{MN89}]
If $(G,c)$ be a zero-demand graph, then there exists a flow in $(G,c)$ if and only if the directed dual $ G^*$ contains no negative weight cycle with respect to weights $c$.
\end{lemma}
\begin{algorithm}
\SetAlgoLined
 \Input{A capacity-demand graph $(G,c,d)$}
 \Promise{$ \sum _v d(v) = 0 $}
 \Output{A pseudo-flow in $(G,c,d)$}
 \BlankLine
 Construct a spanning tree $T$ in $G$\;
 For every edge $ (u,v) \not \in \overleftrightarrow{T}$, set $f'(u,v)=0$\;
 For an edge $ \{u,v\} \in T$, deleting the edge $\{u,v\}$ separates the tree $T$ into two sub-trees, denoted as $ T_u $ (sub-tree containing $u$) and $ T_v $ (sub-tree containing $v$). Set $ f'(u,v)=\sum\limits_{w \in t_u} d(w) $, $\forall_{\{u,v\} \in T}$\;
 \caption{{\tt MN-Pseudo-Flow} \cite{MN89} }
 \label{algo:MN-pseudo-flow}
\end{algorithm}

\begin{algorithm}
\SetAlgoLined
 \Input{A planar bipartite (undirected) graph $ G=(A \cup B,E) $}
 \Output{``Yes'' if $G$ has a perfect matching; ``No'' otherwise}
 \BlankLine
 Construct a capacity-demand graph $(G,c,d)$ as follows: set $ d(u)=1,\; \forall_{u \in A} $ and $ d(v)=-1,\; \forall_{v \in B} $. Also set $ c(u,v)=1 $ and $ c(v,u)=0 $, $\forall_{u \in A, v \in B}$\;
 Construct a pseudo-flow $ f' $ in $(G,c,d)$\;
 Construct a zero-demand graph $ (G, c-f') $\;
 Output Yes if the directed dual $ G^*$ has no negative weight cycle with respect to weights $ (c-f') $; Output No otherwise\;
 \caption{{\tt MN-Decision} \cite{MN89}}
 \label{algo:MN-decision}
\end{algorithm}

\begin{algorithm}
\SetAlgoLined
 \Input{A planar bipartite (undirected) graph $ G=(A \cup B,E) $} 
 \Promise{the directed dual $ G^*=(V^*,E^*)$ has no negative weight cycle with respect to weights $ (c-f') $}
 \Output{A perfect matching in $G$}
 \BlankLine
 Fix a vertex $ s^* \in V^*$\;
 Set $ f''(u^*,v^*) : = dist^w_{G^*} (s^*,v^*) - dist^w_{G^*}(s^*,u^*),\;\forall_{u^*\in V^*}  $\;
 Set $ f=f'' + f' $\;
 For $ u \in A,\: v \in B $ output ``$u$ is matched with $v$'' if and only if $ f(u,v)=1 $\;
 \caption{{\tt MN-Construction} \cite{MN89}}
 \label{algo:MN-construction}
\end{algorithm}

\begin{proof}[Proof of Theorem \ref{thm:perfectmatching}(a)]
We can construct pseudo-flow in a zero-demand graph in log space using {\tt MN-Pseudo-Flow}. Now using the Algorithm \ref{algo:MN-decision}, we reduce {\PM} (Decision) problem in planar bipartite (undirected) graph $G$ to the problem of detecting negative weight cycle in the directed planar graph $ G ^*$, which can be solved in polynomial time and $ O(n^{\frac{1}{2}+\epsilon}) $ space (by Corollary \ref{cor:negcycle}).\\
Now observe in the Algorithm \ref{algo:MN-construction} that the construction of perfect matching in $G$ boils down to the problem of finding the shortest distance $ dist^w_{G^*}(u,v) $ between two vertices $u$ and $v$ in $G^*$ and we can do this again in polynomial time and $ O(n^{\frac{1}{2}+\epsilon}) $ space (by Theorem \ref{thm:shortpath}).
\end{proof}

\subsection{Constructing a Hall Obstacle}
According to the Hall's Theorem \cite{LP86}, a bipartite (undirected) graph $ G=(A \cup B,E) $ has a perfect matching if and only if $ |A|=|B| $ and for every $ S \subseteq A $, $ |N(S)| \ge |S| $, where $ N(S):=\{v \in B | \exists u \in A:(u,v) \in E \} $. A \emph{hall-obstacle} in a bipartite graph $ G=(A\cup B,E) $ is a set $ S \subseteq A $ such that $ |N(S)|<|S| $. We consider the following problems.\\
(i) \emph{{\HO}} (Decision): given a bipartite (undirected) graph $G$, decide whether it contains a Hall-obstacle and\\
(ii) \emph{{\HO}} (Construction): given a bipartite (undirected) graph $G$, construct a Hall-obstacle (if exists).\\
In this subsection, we mention the correspondence between the problem of constructing Hall-obstacle in a planar bipartite graph and the the problem of finding negative weight cycle in a planar graph. For this, we restate some facts from \cite{DGKT12}.\\
Let $ G=(A \cup B,E) $ be a planar bipartite (undirected) graph. Now consider the capacity-demand graph $(G,c,d)$ and a pseudo-flow $f'$ in it as defined in Algorithm \ref{algo:MN-decision}. Let $ C^* $ be a negative weight cycle in the directed dual $ G^*$ with respect to weight $ c-f' $. Let $ (V_1 = A_1 \cup B_1, V_2=A_2 \cup B_2) $ be the cut in $ G $ corresponding to $ C^* $, where $ V_1 $ corresponds to the set of faces of $ G^*$ that are in the interior of $ C^* $ i.e., the vertices of $ G $ that are on one side of the cut corresponding to $ C^* $ and $V_2$ corresponds to the set of faces of $ G^*$ that are in the exterior of $ C^* $ i.e., the vertices of $ G $ that are on the other side of the cut corresponding to $ C^* $. Since $ f' $ is skew-symmetric, $ f'(C^*) $ decomposes into the sum of $ f' $s of the faces(in $ G^*$) that are in the interior of $ C^* $. Thus we have, $$ f'(C^*)=|A_1|-|B_1| \qquad\qquad\quad\quad\quad (1) $$
\begin{lemma}[\cite{DGKT12}]
For the edges $ (a,b) \in (V_1,V_2) $ where $ a\in A_1,\: b\in B_2 $ and $ c(a,b)=1 $, moving $b$ from $ B_2 $ to $ B_1 $ does not increase the weight of the cut and the corresponding cycle in the dual with respect to weights $ c-f' $.
\end{lemma}
\begin{corollary}[\cite{DGKT12}]
$ G^*$ has a negative weight cycle with respect to the weights $ c-f' $ if and only if there exists one such negative weight cycle with respect to the weights $ -f' $, and hence if and only if it has a negative weight cycle with respect to the weights $ c n^4 - f' $.
\end{corollary}
\begin{proof}[of Theorem \ref{thm:perfectmatching}(b)]
By Corollary \ref{cor:negcycle}, we can find negative weight cycle in $ G^*$ with respect to the weights $ cn^4 - f' $ in polynomial time and $ O(n^{\frac{1}{2}+\epsilon}) $ space. Since $ N(A_1) \subseteq B_1 $ and $ |A_1| > |B_1| $, so the set $ A_1 $ forms a Hall-obstacle for $G$ and this completes the proof.
\end{proof}

\subsection{Deciding Even Perfect Matching}
\emph{{\EPM}} denotes the following problem: given a graph $G$ with each edge colored with either Red or Blue, decide whether there exists a perfect matching containing even number of Red edges. Now consider the {\EPM} problem in planar bipartite (undirected) graphs.
\begin{proof}[Proof of Theorem \ref{thm:evenPM}]
Given a planar bipartite (undirected) graph $G=(V,E)$, first construct a perfect matching $M$ in it, which can be done in polynomial time and $ O(n^{\frac{1}{2}+\epsilon}) $ space (by Theorem \ref{thm:perfectmatching}(a)). If $M$ contains even number of red edges, then there is nothing to do. Otherwise, construct an weighted directed graph $ H $ with weight function $w$, as follows: $ H $ contains an edge $(u,v)$ if and only if $ \exists x \in V $ such that $ \{u,x\} \in M $ but $ \{x,v\} \not \in M $. If the matching edge $ \{u,x\} $ and the non-matching edge $ \{x,v\} $ are of the same color, then set $w(u,v)=0$ to $0$; otherwise set $w(u,v)=1$.
\begin{claim}
There exists a perfect matching in $G$ consisting of even number of red edges if and only if $ H $ contains an odd-weight cycle.
\end{claim}
\begin{proof}
Suppose $ H $ contains an odd-weight cycle. Now consider the corresponding portion of the graph in $G$, which is an even length cycle $C$ consisting of alternating matched edge and non-matched edges. Now if we consider a new matching where every non-matched edge in $C$ becomes matched and vice versa, and the new perfect matching (the new matching is perfect as it will not affect the other part of matching in $M$ and also matches every vertex in $C$) with even number of red edges (as there are odd number of pair $ \{u,x\},\{x,v\} $ such that $ \{u,x\} \in M $ but $ \{x,v\} \not \in M $ and they are of different colors).\\
For the reverse direction, lets assume that $ M' $ be a perfect matching in $G$ consisting of even number of red edges. If $ M \cap M' \ne \phi $, then discard those common edges and now consider the sub-graph of $G$, say $ G' $, which contains a edge e if either $ e \in M $ or $ e \in M' $ but $ e \not \in (M \cap M') $. Now vertices in each of the connected component of $ G' $ are of degree $2$ and thus each connected component is just a cycle. Now as $M$ contains odd number of red edges and $ M' $ contains even number of red edges, so at least one of the cycles in $ G' $ contains odd number of red edges and if we consider the corresponding cycle in $ H $ (every cycle in $ G' $ corresponds to one cycle in $ H $), then it must be of odd weight.
\end{proof}

Now observe that the process of construction of $ H $ is nothing but contraction of matched edges present in $M$ on $G$ and as $G$ is a planar graph so we the directed graph $ H $ is also a planar graph. To check the presence of odd-weight cycle in $ H $, we construct another directed graph $ H' $ from $ H $ using the following process: replace every edge $(x,y)$ having weight $0$ by two edges $ (x,v_{xy}) $ and $ (v_{xy},y) $ each with weight $1$. It is easy to see that as $ H $ is a directed planar graph, so is the graph $ H' $.
\begin{claim}
$ H $ contains an odd-weight cycle if and only if $ H' $ contains an odd-weight cycle.
\end{claim}
\begin{proof}
If a cycle in $ H $ uses an edge $(x,y)$ of weight $0$, then there will be a corresponding cycle which will contain the portion $ x \longrightarrow v_{xy} \longrightarrow y $ and as both the edges $ (x,v_{xy}) $ and $ (v_{xy},y) $ have weight $1$, so the parity of the resulting cycle will not change. So, if $ H $ contains an odd-weight cycle then $ H' $ also contains an odd-weight cycle.\\
Similar type of argument can be used to prove the reverse direction as replacing the edges $ (x,v_{xy}) $ and $ (v_{xy},y) $ in $ H' $ by the single edge $(x,y)$ of weight $0$ will result in a cycle of same parity in the graph $ H $.
\end{proof}
As the new graph $ H' $ contains each edge of weight $1$, so we can view this graph as an unweighted directed graph and then we can check the presence of odd length cycle in polynomial time and $ O(n^{\frac{1}{2}+\epsilon}) $ space (by Theorem \ref{thm:oddccycle}) and this completes the proof.
\end{proof}

\section*{Acknowledgement}
The first author would like to thank Surender Baswana for some helpful discussions and comments.

\bibliographystyle{alpha}
\bibliography{references}
\newpage


\newpage

\end{document}